\documentclass[11pt]{article}

\usepackage{amsthm,amssymb,amsmath,amstext,amsfonts}
\usepackage{graphicx}
\usepackage{cleveref}
\usepackage{todonotes}
\usepackage{xspace}
\usepackage{authblk}
\newcommand{\chain}{\textsc{chain}\xspace}
\newcommand{\expchain}{\textsc{expandedchain}\xspace}

\newcommand{\fullpath}{\textsc{fullpath}\xspace}

\title{\bf No truthful mechanism can be better than $n$ approximate for two natural problems
	}

\author[1]{Stefano Leucci} 
\author[2]{Akaki Mamageishvili} 
\author[1]{Paolo Penna}
\affil[1]{Department of Computer Science, ETH Z\"urich}
\affil[2]{Department of Management, Technology and Economics, ETH Z\"urich}

\date{}

\newtheorem{definition}{Definition}
\newtheorem{prop}[definition]{Proposition}
\newtheorem{lemma}[definition]{Lemma}
\newtheorem{claim}[definition]{Claim}
\newtheorem{theorem}[definition]{Theorem}

\newtheorem{remark}{Remark}

\graphicspath{{figures/}}
\begin{document}
\maketitle
\begin{abstract}
 This work gives the first natural non-utilitarian problems for which the trivial \emph{$n$ approximation} via VCG mechanisms is the \emph{best possible}. That is, no truthful mechanism can be better than $n$ approximate, where $n$ is the number of agents. The problems we study are the min-max variant of the \emph{shortest path} and the \emph{(directed) minimum spanning tree} mechanism design problems. In these procurement auctions,  agents own the edges of a network, and the corresponding edge costs are private.  Instead of the total weight  of the subnetwork, in the min-max variant we aim to minimize the maximum agent cost.\end{abstract}

\section{Introduction}
One of the central issues in algorithmic mechanism design concerns the interplay between \emph{optimization} and \emph{incentives}. Roughly speaking, one would like to compute a solution which optimizes a function that \emph{depends on some private information} held by the agents. 
In general, agents may find it convenient to misreport this information, and therefore optimization becomes a critical issue. To overcome this problem, one should design a \emph{truthful} mechanism, that is, a combination of an algorithm and a suitable payment rule such that truth-telling is a \emph{dominant strategy} for all agents.\footnote{Throughout this work we assume the standard \emph{quasi-linear} utilities, meaning that each agent's utility is equal to the difference between the  payment received and the private cost associated to the chosen outcome.}

	In their seminal paper,  Nisan and Ronen \cite{nisan-ronen} considered how well truthful mechanisms can solve a given optimization problem.  They view a mechanism as an abstraction of a ``distributed protocol'' which involves various self-interested parties (agents), a typical scenario in Internet applications. The protocol is required to optimize some function, but also needs to provide suitable incentives to the agents to make sure that they cannot profit from manipulating the protocol. 
	Consider the following simple problem: 

\begin{paragraph}{\emph{Shortest path (auction)} \cite{nisan-ronen}.}
	In a communication network, we would like to establish a path between two distinguished nodes having minimal total length (sum of the costs of the links forming this path). Each link of the network is controlled by an entity (agent) who has a \emph{private} cost for having her link used to connect the two nodes. Without any compensation (payment), agents have an incentive to misreport their own cost (i.e., report a very high cost so that their link is not selected). 
\end{paragraph}

\bigskip
Problems like the one above admit a \emph{truthful} mechanism through the standard VCG construction. This may not always be the case. Indeed,  Nisan and Ronen \cite{nisan-ronen} identified two types of optimization problems:
\begin{itemize}
	\item \emph{Utilitarian (min-sum) problems.} These are the problems where the goal is to  minimize the \emph{sum} of all agents' costs. 
	\item \emph{Non-utilitarian (min-max) problems.} Here the goal is to minimize the \emph{maximum} cost incurred by the agents, as opposed to the sum.
\end{itemize}
They showed that, while all utilitarian problems admit an \emph{exact} mechanism using the standard VCG construction, there are simple and natural  non-utilitarian (min-max) problems for which \emph{no truthful mechanism can guarantee the optimum}. Specifically, they considered the following min-max problem:
\begin{paragraph}{\emph{Unrelated machines scheduling} \cite{nisan-ronen}.}
	We have a set of jobs to be scheduled on $n$ unrelated machines (agents).  Each machine has a type which specifies the processing time (cost) for each of the jobs on this machine. These costs are \emph{private} (known to machine $i$ only), and an allocation of the jobs to machines determines the completion time of that machine (sum of processing times of allocated jobs).  The goal is to return an allocation minimizing the \emph{makespan}, that is, minimizing the \emph{maximum} cost among all machines. 
\end{paragraph}

\bigskip
They showed that even for just \emph{two machines}, no truthful mechanism can be better than $2$-approximate\footnote{Here $c$-approximation means that the mechanism returns an allocation whose makespan is at most $c$ times the optimal makespan for the given input (reported costs).}, while  a trivial $n$-approximation can be obtained via the VCG mechanism,\footnote{Their MinWork mechanism is the VCG mechanism minimizing the sum of all agents costs. This implies that every job is allocated to the fastest machine for that job, which turns out to be an $n$-approximation of the optimal makespan.} and thus the case of two machines is tight. They conjecture that the trivial upper bound is the correct answer for this problem:
\begin{quote}
	\textbf{\emph{Conjecture (Nisan-Ronen):}} \emph{No truthful mechanism for unrelated machines scheduling can have an approximation ratio that is smaller than $n$. }
\end{quote}	
This quite natural and well studied optimization problem suggests that \emph{incentives} do have a \emph{negative impact} on the performance guarantee: Despite being NP-hard, one can compute arbitrarily good approximate solutions in polynomial time. In contrast, no truthful mechanism can be better than  $2$-approximate, even for two machines and even if running in exponential time. Unfortunately, the exact  efficiency loss is still unclear, as the conjecture above is still \emph{open} even for \emph{three machines}, with a large gap between the upper and the lower bound (see related work below):
 \begin{enumerate}
 	\item[(1)] There is a trivial \emph{$n$-approximation} using the VCG mechanism, while the best known lower bound is only a small \emph{constant}.  
 	\item[(2)] The conjecture holds if one makes \emph{additional assumptions} on the class of mechanisms.
 \end{enumerate}
 Though the latter result  supports the conjecture above (showing that natural mechanisms cannot improve VCG for this problems), the following basic question remains open:
 \begin{quote}
 	\emph{How much is lost because of truthfulness?}
 \end{quote}
  Interestingly, similar  state of the art holds for analogous non-utilitarian (min-max) problems. 
 
\subsection{Our contribution}
 We show that the following very simple and natural problems do not admit any truthful mechanism whose approximation is better than $n$, where $n$ is the number of agents:
 \begin{paragraph}{\emph{Min-max path}.}
 	We are given a weighted graph and two distinguished nodes $s$ and $t$, and we would like to find the path connecting them which minimizes the \emph{maximum} cost over the agents. Here the cost of an agent is equal to the sum of the weights of her selected edges, i.e., her share of the path's weight. This agent cost is the same as in the shortest path (auction) problem \cite{nisan-ronen} when the agents have several edges. 
 \end{paragraph}
  
  \begin{paragraph}{\emph{Min-max directed MST}.}
  	We are given a directed weighted graph and one distinguished node $s$, and we would like to 
  	find the directed spanning tree rooted in $s$ which minimizes the \emph{maximum} cost over the agents, where the cost of an agent is the same as in the previous problem. 
  \end{paragraph}

\bigskip
We  prove the following two negative results about the approximability that  truthful mechanisms can achieve for the  problems above: 

\begin{theorem}\label{th:inapxpath}
	No truthful mechanism for the min-max path problem can be better than $n$-approximate. 
\end{theorem}

\begin{theorem}\label{th:inapxMST}
	No truthful mechanism for the min-max directed MST problem can be better than $n$-approximate. 
\end{theorem}

To the best of our knowledge, these are the first examples for which there is a \emph{strong separation} between truthful approximations and non-truthful ones. 
Note that these results hold without any additional assumption on the mechanism (including its running time). This gives an unconditional lower bound based solely on truthfulness.  A trivial $n$-approximation can be obtained by simply running the VCG mechanism: A shortest path (respectively minimum spanning tree) is an $n$-approximation of the min-max path (respectively, min-max directed MST). Thus $n$ is a \emph{tight} bound for truthful mechanisms in both problems.

The merit of this result is the simplicity of the problem and of the proof. Specifically, in both problems, the agents' costs for a solution are the same as in the shortest-path problem with agents owning several edges (sum of the costs of chosen edges). In the auction terminology, we are in a simple case of \emph{no externalities}, meaning that agents care only about the items that they get (which of their edges are chosen).  
Many of the interesting problems are of this sort, and the main difficulty here is that one cannot invoke Roberts Theorem saying that truthfulness implies that the mechanism must be an affine maximizer (in the VCG family).

As we discuss in the next section, all prior inapproximability results for other non-utilitarian min-max problems either (1) are significantly  weaker as only a small \emph{constant} factor on the inapproximability is known or (2) they only apply to certain classes of mechanisms as they are based on additional assumptions.
As we discuss in Section~\ref{sec:conclusion},  our problems can be thought of as a generalization of the unrelated machines scheduling problem.

 \subsection{Related work}
 Arguably, the most general mechanism design technique is  VCG and its simple generalization known as \emph{affine maximizers} (``weighted VCG''). Roughly speaking, these mechanisms  maximize the weighted social welfare (sum of agents valuations) over a fixed subset of the possible solutions. Equivalently, for problems involving private costs, these mechanisms minimize the weighted social cost (sum of agents costs). These constructions are general in the sense that they require no assumption on the agents' domain, that is, they can be applied to \emph{unrestricted domains}.\footnote{This is the case where, for a finite set $\mathcal A$ of alternatives or outcomes, each agent $i$ has a valuation which can be \emph{any} function $v_i: \mathcal A \rightarrow \Re$.} A famous result by Roberts \cite{Roberts} says that affine maximizers are   \emph{the only} truthful mechanism on unrestricted domains.
 
For most of the interesting problems, truthful mechanisms other than VCG may still be possible
 because the problem deals inherently with a \emph{restricted} domain. 
 A classical example is the \emph{no externalities} condition in combinatorial auctions which says that the agents valuations depend uniquely on the items they get, and not on who gets the other items. For multi-unit auctions  there are truthful  \emph{non-VCG}  mechanisms that outperform any VCG mechanism for the same problem \cite{DobNis15}. 
 In a \emph{one-dimensional} or \emph{single-parameter} setting, valuations are linear in the number of items allocated (the private parameter being the valuation for a single unit). These problems are usually easier and truthfulness is less stringent than in the multi-dimensional case (see e.g. Myerson~\cite{Mye81} and Archer and Tardos~\cite{ArcTar01}). 
Truthfulness can be characterized by a so called \emph{monotonicity} condition (see e.g., Rochet~\cite{Roc87}, Bikhchandani et al.~\cite{bikhchandani2006weak}, 
Saks and Yu~\cite{SakYu05}). Intuitively, this is a property of the algorithm (allocation rule) and it prescribes how the allocation of an agent should change if only this agent changes her reported type. In that sense, this condition is \emph{local} as it focuses on one agent at a time. Lavi et al. \cite{Lavi2008} presented an alternative proof of Roberts' theorem using monotonicity (and an extra condition called ``player decisiveness'').
 
 All known lower bounds for the unrelated machines scheduling (and others) are based on the above mentioned monotonicity condition. 
 	The difficulty in the unrelated machine scheduling problem is that its domain is neither one-dimensional  nor unrestricted.  
 The problem becomes interesting already for $n=3$ machines, and the gap between the best upper bound and the best lower bound gets wider as $n$ grows. Nisan and Ronen \cite{nisan-ronen} proved that $n$-approximation can be obtained via the VCG mechanism, and that for $n=2$ no mechanism can be better than $2$-approximate.  This has triggered a fairly large number of papers that focused on this problem and some variants \cite{3-machines, best_lower_bound,anonymous,random-mechanism,fractional-scheduling}. Christodoulou et al.~\cite{3-machines}  showed a lower bound of  $1+\sqrt{2}\simeq 2.41$   for $n=3$ machines.  Koutsoupias and Vidali~\cite{best_lower_bound} proved a lower bound of $1+\varphi\simeq 2.618$ for arbitrarily many machines (i.e., for $n \rightarrow \infty$).  These are the best known lower bounds for this problem, and stronger lower bounds have been obtained only by making \emph{additional assumptions} on the mechanism. Specifically, a \emph{lower bound of $n$} has been obtained under the following various assumptions:  Nisan and Ronen~\cite{nisan-ronen} considered mechanisms whose payments are \emph{additive} in the jobs; Mu'alem and Schapira~\cite{random-mechanism} considered \emph{strong monotonicity}, a stronger condition than the one characterizing truthfulness; Ashlagi et al.~\cite{anonymous} focused on \emph{anonymous} mechanisms, i.e., mechanisms whose allocation does not depend on the names of the agents.
 
 Christodoulou et al.~\cite{fractional-scheduling} studied the \emph{fractional} version of the unrelated machines problem. They proved that still there is a lower bound of $2-\frac{1}{n}$ for this seemingly simpler problem, and obtained a slightly better upper bound of $\frac{n+1}{2}$ compared to the original problem. Similarly to the original problem, if one makes the additional assumption on the mechanism, namely that the algorithm is \emph{task independent}\footnote{This condition requires that the allocation of a task depends only on the processing times of this task on the machines, and not on the other tasks' processing times.}, then a matching lower bound of $\frac{n+1}{2}$ holds.  
 
 Another problem which exhibits a similar structure to the unrelated machines is the \emph{inter-domain routing} problem by  
Mu'alem and Shapira~\cite{random-mechanism}. Here, we have a graph whose nodes are the agents, and the solutions are the trees directed towards a destination node; each solution determines the amount of traffic that each  of node $i$ receives from its neighbors; each node  $i$ has a per-unit cost for each of the neighbors, and the goal is to minimize the \emph{maximum} cost among the nodes.  Again, we have a multi-dimensional non-utilitarian (min-max) problem for which  the best lower bound is a small \emph{constant} and the best upper bound is $n$ using the VCG mechanism \cite{random-mechanism}. Gamzu~\cite{Gam11} proved a lower bound of $2$, which is the best known lower bound for this problem and it also applies to  \emph{randomized} mechanisms.

 The only significant improvements on the upper bounds have been obtained for  \emph{single-parameter} or \emph{two-values} domains.
 Archer and Tardos~\cite{ArcTar01} showed that, unlike in the unrelated machines, the  \emph{related machines} case admits an \emph{exact} (exponential time) truthful mechanism, and a constant approximation can be obtained by a randomized polynomial-time truthful mechanism.  Christodoulou and Kov\'acs \cite{truthfulPTASrelated} even obtained a \emph{polynomial-time} deterministic truthful approximation scheme, that is, for any $\epsilon>0$ there is a truthful polynomial-time mechanism which computes a $(1+\epsilon)$-approximate solution. Mu'alem and Shapira~\cite{random-mechanism} observed that the \emph{single-parameter} version of their inter-domain routing problem also admits an exact truthful mechanism (as opposed to the multi-dimensional version). 
Lavi and Swamy \cite{lavi2009truthful} gave a truthful \emph{$3$-approximation} mechanism for the unrelated machines restricted to \emph{two-values} domains, i.e., when the cost of executing a job on a machine can be only ``low'' or ``high''. Their proof uses (in a non-trivial way) the cycle monotonicity condition by Rochet~\cite{Roc87}. Yu \cite{yu2009truthful} extended the result to \emph{two-range values} where the costs belong to two ranges which are ``sufficiently far apart''.

 	Several papers considered the power of \emph{randomization} for this problem, essentially showing that the situation is not much different. In particular, Mu'alem and Shapira~\cite{random-mechanism} proved a lower bound of $2-\frac{1}{n}$, meaning that it is still  impossible to achieve exact solutions (note that this is the same lower bound for the fractional version).  Lu and Yu~\cite{LuYu08a,LuYu08b} gave a randomized mechanism with approximation  $0.8368n$ and  $\frac{n+5}{2}$, depending on whether we consider \emph{universally truthful} or \emph{truthful in expectation} mechanisms. Intuitively,  the former is a stronger requirement which says  that truth-telling is a dominant strategy, even if agents would  know the random bits, while the latter needs agents to care about their expected utility. For a separation result in a subclass of the two-values domains in unrelated machines see Auletta et al.~\cite{auletta2015mechanisms}. 
 	Like for deterministic mechanisms, the known upper bounds for randomized mechanisms are optimal if one restricts to \emph{task-independent} mechanisms ($\frac{n+1}{2}$ is a lower bound \cite{LuYu08a}).

 Kov\'acs and Vidali \cite{KovVid15} studied mechanisms that satisfy \emph{strong monotonicity} for the unrelated machines scheduling domain. They  provide characterizations of mechanisms satisfying also some additional requirements for the case of \emph{two jobs}. 
Christodoulou et al.~\cite{ChrVid08} characterized the class of truthful mechanisms for \emph{two machines}. In both cases, the resulting class is a \emph{wider} class than  VCG/affine minimizers, which suggests that trying to extend Roberts' theorem to the  scheduling domain in order to prove a lower bound might be hopeless. Indeed, only for the case of \emph{two jobs  and  two machines} truthful mechanisms coincide with VCG mechanisms/affine minimizers  \cite{ChrVid08}. 
 	
 Chawla et al.~\cite{chawla2013prior} and Giannakopoulos and Kyropoulou~\cite{Giannakopoulos2015} considered a   \emph{Bayesian} setting where the agents types  are drawn according to a certain probability  distribution.

 \section{The two problems and  truthfulness}
Both problems we consider share the following features:
 \begin{itemize}
 	\item We are given a weighted (directed or undirected) graph $G=(V,E)$.
 	\item The edges are partitioned among $n$ agents $1, \dots, n$,
 	where each agent $i$ owns a subset $E_i$ of edges, and each edge $e$ belongs to exactly one agent. 
 	\item The cost (weight) of edge $e$ is denoted by
 	\[
 	t_{i}(e)
 	\]
 	where $t_i$ is the \emph{type} of agent $i$ owning this edge. The type of agent $i$ is \emph{private knowledge}, and each agent $i$ can report a possibly different type.
 	\item We have a set $X$ of feasible solutions, where each feasible solution $x \in X$ consists of a  subset of edges. Any feasible solution $x$  costs agent $i$ the sum of the weights (costs) of her edges included in the solution,
 	\[
 	t_i(x) = \sum_{e \in x\cap E_i} t_i(e)\ .
 	\]
 	\item The goal is to find a feasible solution $x^*$ minimizing the maximum agent cost $cost(x,t)=\max_i t_i(x)$ where $t = (t_i)_{i=1,\dots,n}$, that is, $x^* \in \arg \min_{x \in X} cost(x,t)$.
 \end{itemize}
 We call problems with the above structure  \emph{graph problems}, regardless of the last item (the optimization criteria).
The main differences between the two problems we consider are: (i) whether the graph is undirected or not, and (ii)  the structure of the set $X$ of feasible solutions (i.e., whether a solution $x \in X$ corresponds to a path or a tree in $G$). 

\begin{definition}[min-max path and min-max directed MST]\label{def:problems}
	In the \emph{min-max path} problem, the graph is undirected, and the set of feasible solutions consists of all paths connecting two given nodes. In the \emph{min-max directed minimum spanning tree (MST)}, the graph is directed, and the feasible solutions consist of all directed trees (i.e., arborescences) connecting a given vertex (the root) to all other nodes (there is a directed path from the root to every other node).
\end{definition}

We conclude with the formal definition of a \emph{truthful mechanism}.  
 \begin{definition}[truthful mechanism]
 	A mechanism $(A,P)$ is \emph{truthful} if truth-telling is a dominant strategy (utility maximizing) for all agents. That is, for any vector $\tilde t=(\tilde t_1,\ldots,\tilde t_n)$ of costs reported by the agents, for any $i$, and for any true cost $t_i$ of agent $i$,
 	$$P_i(t) - t_i(x) \geq P_i(\tilde t) - t_i(\tilde x),$$
 	where $t=(\tilde t_1,\ldots,\tilde t_{i-1},t_i,\tilde t_{i+1},\ldots,\tilde t_n)$, $x=A(t)$, and $\tilde x=A(\tilde t)$.
 \end{definition}
 
 \subsection{Implications of truthfulness}
In this section, we state the main basic properties that any truthful mechanism must satisfy for our problems.

 \begin{definition}[monotone algorithm]
     Algorithm $A$ is monotone if, for any $t$, for any $i$, and for any $t'_i$, 
     \begin{align}\label{eq:mon}
         t_i(x) + t_i'(x') \leq t_i(x') + t_i'(x) 
     \end{align}
     where $x=A(t)$ and $x'=A(t_i',t_{-i})$.
 \end{definition}

The following is a well-known result (see e.g., \cite{bikhchandani2006weak,random-mechanism}).
\begin{prop}\label{prop:truthful-monotone}
	If a mechanism $(A,P)$ is truthful then $A$ is monotone.
\end{prop}
From the above property, one can easily derive the following condition that must be satisfied by any truthful mechanism. 
 
 \begin{lemma}\label{le:mon}
 Let $(A,P)$ be a truthful mechanism for a graph problem.
 Let $t$ be a vector of reported types and let $i \in \{ 1, \dots, n \}$ be an agent. 
  For $t$ the mechanism selects a subset $A_i(t)$ of the edges owned by $i$. Now consider types in which these edges are less costly, while all non-selected edges of $i$ are more costly:
 \begin{align}
 t'_{i}(e)<& t_{i}(e) & \text{for } e \in A_i(t) \enspace , \label{eq:mon-edges-up}\\
t'_{i}(e)>& t_{i}(e) & \text{for } e \not\in A_i(t) \enspace . \label{eq:mon-edges-down}
 \end{align}
Then the mechanism must select the same subset of edges of $i$, provided all other agents' types are unchanged, i.e., 
\begin{align}\label{eq:stable}
A_i(t)=A_i(t_i',t_{-i}) \enspace .
\end{align}
 \end{lemma}
 \begin{proof}
 	Consider the symmetric difference $S \cup S'$ between the edges owned by $i$ selected in the two solutions, where $S := A_i(t)\setminus A_i(t')$ and $S':=A_i(t')\setminus A_i(t)$. We shall prove that both $S$ and $S'$ must be empty. If $S$ is non-empty, then  \eqref{eq:mon-edges-up} implies
 	\begin{align}
 		t_i(S)>t_i'(S) \ .
 	\end{align}
 	Similarly, if $S'$ is non-empty, then \eqref{eq:mon-edges-down} implies 
 	\begin{align}
 	t'_i(S')>t_i(S') \ .
 	\end{align}
 	We thus have $t_i(x) - t_i(x') + t_i'(x') - t_i'(x) = t_i(S)-t_i(S')+t_i'(S')-t_i'(S)>0$ if at least one of $S$ and $S'$ is non-empty. This contradicts the monotonicity condition \eqref{eq:mon}, and thus $(A,P)$ cannot be truthful (\Cref{prop:truthful-monotone}). 
 \end{proof}

 \subsection{Min-max path (Proof of \Cref{th:inapxpath})}

 \begin{figure}
 	\centering
 	\includegraphics[scale=.7]{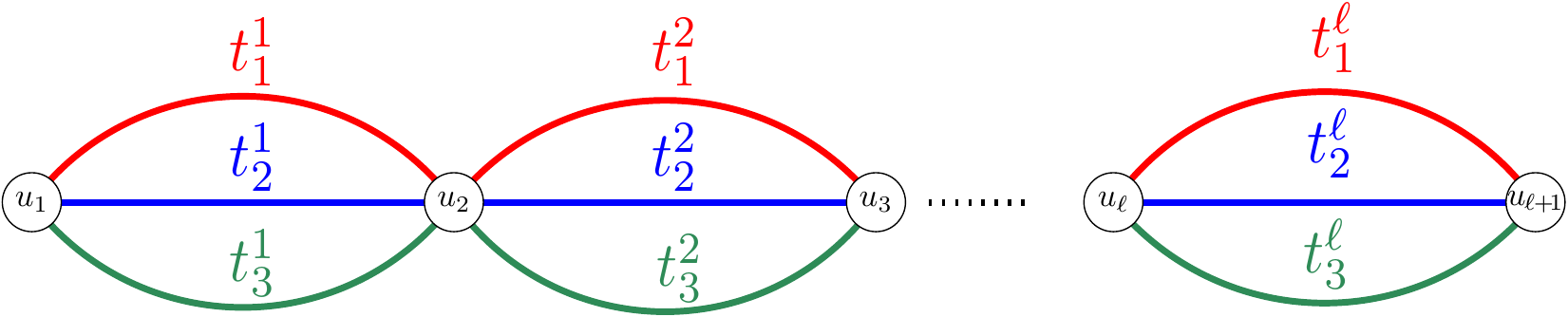}
 	\caption{The \chain graph for $n=3$ agents.}
 	\label{fig:chain-graph}
 \end{figure}
 We construct a graph \chain as follows (see Figure~\ref{fig:chain-graph}). 
 Consider the  concatenation of $\ell$ \emph{blocks} $B_1,\ldots, B_\ell$, where each block $B_k$ consists of $n$ parallel edges connecting $u_k$ to $u_{k+1}$, one for each agent. We denote by $t_i^k$ the cost of the edge $e_i^k$ of agent $i$ in block $k$. We shall specify the number $\ell$ of blocks and the costs $t_{i}^k$ later in the proof. 
 
 We now transform \chain into another graph \expchain in the following way (see Figure~\ref{fig:chain-into:expandedchain}). In every block $B_k$, 
each edge $e_i^k$ of agent $i$ is replaced by a path $e_{i,1}^k,\ldots,e_{i,n}^k$ of $n$ edges.  Edge $e_{i,j}^k$ belongs to agent $j$ and the costs are as follows. Edge $e_{i,i}^k$ owned by $i$ has the original cost of $e_i^k$, and all the other $n-1$ edges in this path have a tiny cost $\epsilon$.
 We name the resulting graph \expchain.
\begin{figure}
	\begin{center}	
		\includegraphics[scale=.7]{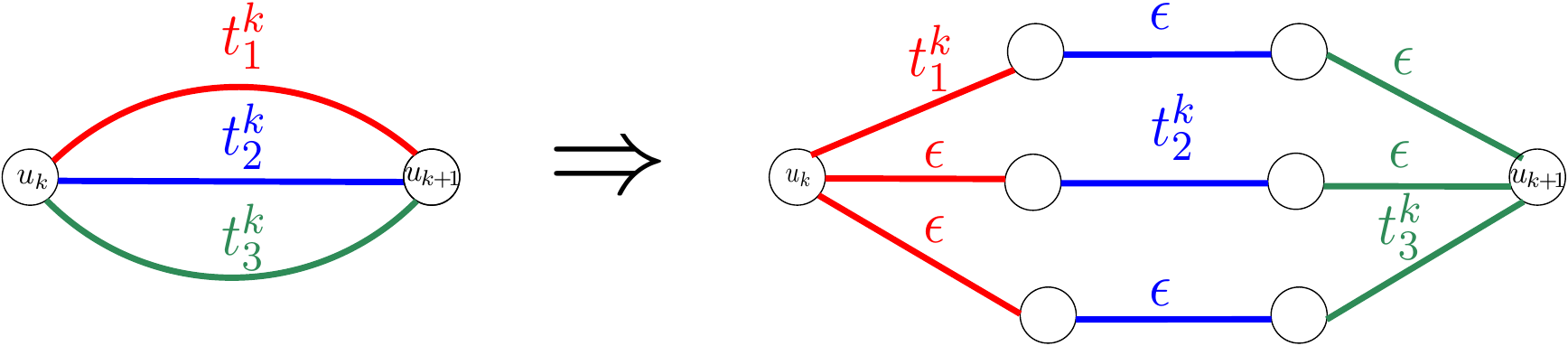}
	\end{center}
	\caption{Transforming \chain into \expchain for $n=3$ agents: Each edge in a block of \chain is replaced by  a path with $n$ edges, one for each agent.}
	\label{fig:chain-into:expandedchain}
\end{figure}

 \begin{lemma}\label{le:mon-stable}
 	Let $(A,P)$ be a truthful mechanism, and let $t$ and $t_i'$ be defined as in \Cref{le:mon}. Then, on \expchain, the algorithm must return the same solution on these two types, i.e.,
 	\[
 	A(t)=A(t_i',t_{-i}) \enspace .
 	\]
 \end{lemma}
 
 \begin{proof}
 	We show that $A(t)$ and $A(t_i',t_{-i})$ must select the same path in each block of \expchain. Fix a block $B_k$ and notice that each of the $n$ paths contains exactly one edge $e_i^k$ from agent $i$. By \Cref{le:mon}, $e_i^k$ must also be selected in  $A(t_i',t_{-i})$, which forces the same path in this block to be selected as well (and no other paths are selected because the solution must be a simple path).
 \end{proof}

 \paragraph{Proof of \Cref{th:inapxpath}.}
 Suppose all costs are equal to $1$ in the \chain graph (before the transformation) and consider the corresponding graph \expchain.
 
	Notice that the algorithm must select exactly one simple path between $u_k$ and $u_{k+1}$.
	Let $\ell_i$ denote the number of these paths in \expchain that contain an edge of cost $1$ owned by agent $i$. Without loss of generality, assume 
 \[
 0 \leq \ell_n \leq \ell_{n-1} \leq \cdots \leq  \ell_2 \leq \ell_1
 \]
 by simply renaming the agents. Then, since all $\ell$ blocks 
 contain one selected path, $$\ell_1 \geq \ell/n.$$
 Now iteratively repeat the following transformation  on the costs $t_{i}$ for agents $i=2,\dots,n$:

 \begin{itemize}
 	\item Set the cost of all edges of agent $i$ selected in the solution to $0$, and increase by $\epsilon$ the cost of edges of $i$ not in the solution.  
 \end{itemize}
 By \Cref{le:mon-stable}, the solution cannot change and therefore, for the final types $t^*$, the cost of agent $1$ is still at least $\ell_1$, thus implying
 \[
 cost(A(t^*),t^*) \ge  \ell_1 \enspace .
 \]
   However, the optimum for this instance $t^*$ would be to distribute these $\ell_1$ paths among the $n$ agents, as each of the other paths contributes at most $\epsilon$ to the cost of the solution:
 \[
 opt(t^*) \leq \left\lceil \frac{\ell_1}{n}\right\rceil (1+\epsilon) + 2\epsilon \ell  \leq \left(\frac{\ell_1}{n} + 1\right)(1+\epsilon) + 2\epsilon \ell \ ,
 \]
 where the term $2\epsilon \ell$ accounts for the fact that a solution always selects, in each of the $\ell$ blocks, at most one edge of initial cost $\epsilon$ per agent, and the new cost is at most $2\epsilon$.  By taking $\epsilon = \frac{1}{2\ell}$ this implies
 \[
	 opt(t^*) \leq \frac{\ell_1}{n} + \frac{\ell_1}{2n \ell} + 1 + \frac{1}{2\ell} + 1  \le  \frac{\ell_1}{n} + 4\ .
 \]
 We thus have
 \[
 \frac{cost(A(t^*),t^*)}{opt(t^*)} \geq  n \cdot \frac{\ell_1}{\ell_1 + 4n}
  = n - \frac{4n^2}{\ell_1 + 4n}
  > n - \frac{4n^2}{\ell_1} \ge n - \frac{4 n^3}{ \ell} \ .
 \]
 Taking $\ell$ arbitrarily large gives a lower bound of $n - \delta$, for any $\delta>0$. \hfill \qed

\section{Min-max directed MST (Proof of \Cref{th:inapxMST})}

The \emph{min-max directed MST} problem is  similar to the min-max path problem with few differences: (i) We have a weighted \emph{directed} graph $G$; (ii) A solution is a directed spanning tree (arborescence) $x$ rooted at some distinguished node $s$.

\subsection{Adapting the reduction} 
In the proof we start from the same chain of blocks as in the proof of \Cref{th:inapxpath}, where edges are now directed from $u_k$ to $u_{k+1}$, but the transformation is slightly different (see Figure~\ref{fig:chain-into-expchain:directed} for the intuition).

\begin{figure}
	\begin{center}
		\includegraphics[scale=.99]{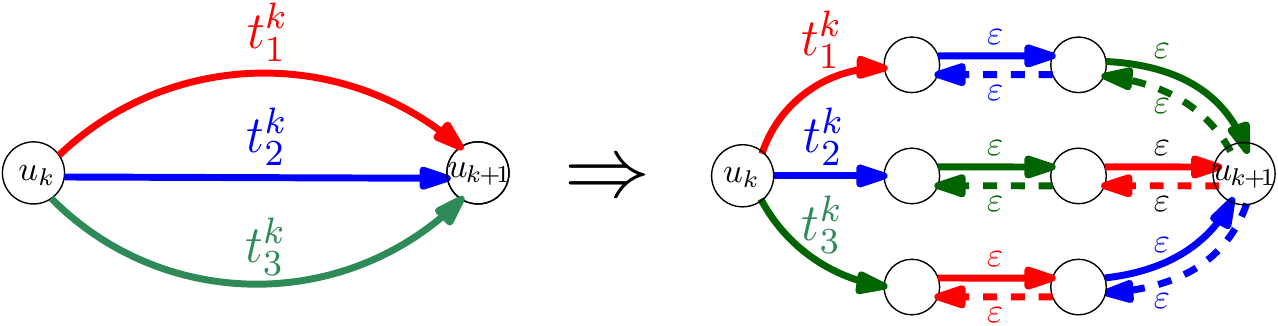}
	\end{center}
	\caption{Adapting the reduction to directed MST for $n=3$ agents.}
\label{fig:chain-into-expchain:directed}
\end{figure}

Specifically, the left block consists of $n$ parallel edges directed from left to right, each of them belonging to a different agent. 
The $i$-th parallel edge in each block is then replaced by a path of $n$ rightward edges plus $n-1$ leftward edges as follows:
\begin{enumerate}
	\item \label{itm:first-edge} The first rightward edge belongs to agent $i$ and its cost is the same as the cost in the left graph (parallel edges).
	\item \label{itm:rightward-edges} The remaining $n-1$ rightward edges belong to agent $(i \mod n) + 1, (i+1 \mod n) + 1, \ldots, (i+n-2 \mod n) + 1$. This order is \emph{the same across all blocks} (this will be crucial in the following). All these edges have cost $\epsilon$. 
	\item \label{itm:lefttward-edges} Each of the previous $n-1$ edges is paired with an edge directed in the opposite direction, also of cost $\epsilon$. 
\end{enumerate}
The whole construction consists of again a chain of $\ell$ blocks, where block $B_k$ consists of the $n$ parallel paths connecting $u_k$ with $u_{k+1}$.
The leftmost node in the first block, i.e., $u_1$, is the root. 

\begin{definition}
	For any block $j$ and any agent $i$, we define $\fullpath(i,j)$ as the unique path of $n$ rightward edges, as defined in Items~\ref{itm:first-edge}-\ref{itm:rightward-edges} above, in which the first edge is owned by agent $i$. We say that the mechanism \emph{selects agent $i$} in block $j$ if $\fullpath(i,j)$  is selected. 
\end{definition}

Note that, in general, it is possible to select more than one agent per block. However, due to the presence of the leftward edges, there  exists always  a feasible solution that selects only \emph{one} agent per block.

\paragraph{Proof of \Cref{th:inapxMST}.}
We are now in a position to apply a similar type of analysis as we did in the proof of \Cref{th:inapxpath}. 
 Initially all costs are equal to $1$ in the directed chain graph (before the transformation). Since the algorithm must select a full path in each block, there must be an agent $i^*$ which has been selected in at least
$\ell_{i^*}\geq \ell/n$ blocks. 

Without loss of generality (as this can be guaranteed by simply renaming the agents), we assume $i^* = 1$.
Notice that $\fullpath(i^*,j)$ traverses the edges owned by agents $1,2,\dots,n$ in left-to-right order.

Now we iteratively repeat the following transformation on the costs $t_{i}$ of agent $i$, where agents are considered in decreasing order from $n$ to $2$: 
 \begin{itemize}
 	\item For the current agent $i$, reduce the cost of all edges of $i$ taken in the solution to $0$, and increase the cost of non-selected edges by $\epsilon$.
 \end{itemize}
 Intuitively, this particular order will ensure that, after each transformation, the mechanism must still select agent $i^* = 1$ in each of the  $\ell_{1}$ blocks where it was initially selected. 
 \begin{claim}\label{cla:directed-MST:stable}
 	For any block $j$ where  $\fullpath(1,j)$ has been initially selected, the following holds. After the $k$-th step of the above transformation, (that is, at the beginning if $k=0$ or after modifying the edge costs of the edges owned by agent $n-k+1$, if $k>0$), the mechanism must still select  all edges in $\fullpath(1,j)$ that belong to agents $1, 2, \dots, n-k$.
 \end{claim}
 \begin{proof}
 	The proof is by induction on $k$, where the base case $k=0$ is trivial since initially the mechanism has selected all edges in $\fullpath(1,j)$. As for the inductive step, suppose that after the $k$-th step the edge of agent $n-k$ is  still selected (inductive hypothesis). The $(k+1)$-th step consists of lowering the costs of  all edges of agent $n-(k+1)+1=n-k$ which are currently selected, while increasing the cost of the edges owned by agent $n-k$ not currently selected. By \Cref{le:mon}, the edge of agent $n-k$ in $\fullpath(1,j)$  must be selected also after the transformation. Since the feasible solution must be a directed spanning tree, this forces all previous edges in $\fullpath(1,j)$  to be selected, i.e., those from agents  $1, 2, \dots, n-k-1$. That is, the claim holds for $k+1$.
 \end{proof}
 The above claim shows that, after the last iteration, we still have agent $i^*=1$ with a cost at least $\ell_{i^*}\geq \ell / n$. However, in the modified costs $t^*$, the optimum is at most 
 \begin{equation}\label{eq:opt-ub:directed-MST}
 opt(t^*) \leq \left\lceil \frac{\ell_{i^*}}{n}\right\rceil (1+\epsilon) + 4\epsilon n\ell\leq \left(\frac{\ell_{i^*}}{n} + 1\right)(1+\epsilon) + 4\epsilon n\ell \enspace . 
 \end{equation}
 This is so because,  in every block $j'$ where the mechanism initially selected some agent  $i' \neq i^*$,   
 $\fullpath(i',j')$ has cost at most $\epsilon$, due to the single edge of $i^*$ whose initial cost was $\epsilon$.  Moreover, each agent has at most $2n\ell$ edges of initial cost $\epsilon$, and after the transformation this cost is at most  $2\epsilon$.  
  Then, the first inequality in \eqref{eq:opt-ub:directed-MST} can  be obtained by considering the solution in which the $\ell_{i^*}$ blocks selecting $i^*$ are redistributed evenly among all agents.

 By taking $\epsilon = \frac{1}{4\ell n}$,  \eqref{eq:opt-ub:directed-MST} implies
  \[
  opt(t^*) \leq \frac{\ell_{i^*}}{n} + \frac{\ell_{i^*}}{4\ell n^2 } + 1 + \frac{1}{4\ell n} + 1  \le  \frac{\ell_{i^*}}{n} + 4\ .
  \]
 The remainder of the proof is identical to that of \Cref{th:inapxpath}. \hfill \qed

\section{Conclusion and open questions}\label{sec:conclusion}
It is worth noticing that our problems can be thought of as a generalization of the unrelated machines scheduling. Indeed, by considering the \chain graph (see Fig.~\ref{fig:chain-graph}), one can think  of each $t_{i}^j$ as the cost to process job $j$ on machine $i$ from the scheduling problem. Any solution for the min-max path problem corresponds to a job scheduling with the same makespan. For the min-max directed MST problem, a similar argument applies since every (directed) tree on this graph is also a path.

An important point in our proofs is the role of the combinatorial structure, in particular how we expand \chain into \expchain. In a nutshell, it allows to control how the solution changes when a single player cost changes (cf \Cref{le:mon-stable}  and \Cref{cla:directed-MST:stable} for the two reductions, respectively). In several proofs of lower bound results for unrelated machines scheduling, this is done by assuming \emph{extra} properties of the mechanism. The lower bound by Mu'alem and Schapira \cite{random-mechanism} is based on \emph{strong monotonicity}, which implies that the algorithm is somehow breaking ties among the solutions in a fixed manner.

The following two natural questions are still open:

\begin{paragraph}{Question~1:}
\emph{Is it possible to extend our lower bounds to \emph{randomized} mechanisms?}
\end{paragraph}

\begin{paragraph}{Question~2:}
\emph{Is it possible to extend the lower bound for min-max directed MST  to the \emph{undirected} case?}
\end{paragraph}
\bigskip

We do not know the answer to either question. The main reason is that certain key properties of our reduction seem difficult to obtain. Regarding the first question, randomized mechanisms for min-max path could be studied by looking at the \emph{fractional} version of the problem, i.e., the one in which we send a unit of \emph{flow} from the source to the destination (this flow can be divided arbitrarily). Now the monotonicity condition on our reduction does not guarantee anymore that the solution does not change under certain conditions (\Cref{le:mon-stable}). 
Regarding the second question, the monotonicity condition of the min-max directed MST problem is ensured by the \emph{direction} of the edges and again it is lost in the undirected case (\Cref{cla:directed-MST:stable}). 

Furthermore, we observe that our reductions use few values and, in fact, they would fall in the class of \emph{two-range values} studied by Yu \cite{yu2009truthful} for unrelated machines scheduling. Since Yu \cite{yu2009truthful} shows that \emph{constant-approximation} is indeed possible in this restriction, we obtain a \emph{separation} between the scheduling problem and our two min-max problems for such restricted domains.

Finally, we remark that  the min-max path problem can be approximated efficiently in a non-truthful manner (thus complexity is not an issue). In particular, there exists a \emph{polynomial-time approximation scheme}\footnote{This means that, for every $\epsilon>0$, there exists a polynomial-time $(1+\epsilon)$-approximation algorithm.} for any constant number of agents. This follows easily from a result by Tsaggouris et al.~\cite{Tsaggouris2009} on the \emph{multi-objective} version of the shortest-path problem (see Appendix~\ref{app:PTAS}) and it is the same result that holds for scheduling a constant number of unrelated machines. Whether the same holds for the min-max MST is an interesting open question.

\bibliographystyle{plain}
\bibliography{bibliography}

\begin{thebibliography}{10}

\bibitem{ArcTar01}
Aaron Archer and {\'E}va Tardos.
\newblock Truthful mechanisms for one-parameter agents.
\newblock In {\em Proc. of the 42nd IEEE Symposium on Foundations of Computer
  Science (FOCS)}, pages 482--491, 2001.

\bibitem{anonymous}
Itai Ashlagi, Shahar Dobzinski, and Ron Lavi.
\newblock Optimal lower bounds for anonymous scheduling mechanisms.
\newblock {\em Math. Oper. Res.}, 37(2):244--258, 2012.

\bibitem{auletta2015mechanisms}
Vincenzo Auletta, George Christodoulou, and Paolo Penna.
\newblock Mechanisms for scheduling with single-bit private values.
\newblock {\em Theory of Computing Systems}, 57(3):523--548, 2015.

\bibitem{bikhchandani2006weak}
Sushil Bikhchandani, Shurojit Chatterji, Ron Lavi, Ahuva Mu'alem, Noam Nisan,
  and Arunava Sen.
\newblock Weak monotonicity characterizes deterministic dominant-strategy
  implementation.
\newblock {\em Econometrica}, pages 1109--1132, 2006.

\bibitem{chawla2013prior}
Shuchi Chawla, Jason~D Hartline, David Malec, and Balasubramanian Sivan.
\newblock Prior-independent mechanisms for scheduling.
\newblock In {\em Proc. of the 45th annual ACM Symposium on Theory of Computing
  (STOC)}, pages 51--60, 2013.

\bibitem{fractional-scheduling}
George Christodoulou, Elias Koutsoupias, and Annam{\'{a}}ria Kov{\'{a}}cs.
\newblock Mechanism design for fractional scheduling on unrelated machines.
\newblock {\em {ACM} Trans. Algorithms}, 6(2):38:1--38:18, 2010.

\bibitem{ChrVid08}
George Christodoulou, Elias Koutsoupias, and Angelina Vidali.
\newblock A characterization of 2-player mechanisms for scheduling.
\newblock In {\em Proc. of the 16th Annual European Symposium on Algorithms
  (ESA)}, volume 5193 of {\em Lecture Notes in Computer Science}, pages
  297--307, 2008.

\bibitem{3-machines}
George Christodoulou, Elias Koutsoupias, and Angelina Vidali.
\newblock A lower bound for scheduling mechanisms.
\newblock {\em Algorithmica}, 55(4):729--740, 2009.

\bibitem{truthfulPTASrelated}
George Christodoulou and Annam{\'a}ria Kov{\'a}cs.
\newblock A deterministic truthful {PTAS} for scheduling related machines.
\newblock {\em SIAM Journal on Computing}, 42(4):1572--1595, 2013.

\bibitem{DobNis15}
Shahar Dobzinski and Noam Nisan.
\newblock {Multi-unit auctions: Beyond Roberts}.
\newblock {\em Journal of Economic Theory}, 156:14--44, 2015.

\bibitem{Gam11}
Iftah Gamzu.
\newblock Improved lower bounds for non-utilitarian truthfulness.
\newblock {\em Theoretical Computer Science}, 412(7):626 -- 632, 2011.

\bibitem{Giannakopoulos2015}
Yiannis Giannakopoulos and Maria Kyropoulou.
\newblock {The VCG Mechanism for Bayesian Scheduling}.
\newblock In {\em Proc. of the 11th International Conference on Web and
  Internet Economics (WINE)}, pages 343--356, Berlin, Heidelberg, 2015.
  Springer Berlin Heidelberg.

\bibitem{best_lower_bound}
Elias Koutsoupias and Angelina Vidali.
\newblock A lower bound of 1+\emph{{\(\varphi\)}} for truthful scheduling
  mechanisms.
\newblock {\em Algorithmica}, 66(1):211--223, 2013.

\bibitem{KovVid15}
Annam\'aria Kov\'acs and Angelina Vidali.
\newblock A characterization of n-player strongly monotone scheduling
  mechanisms.
\newblock In {\em Proc. of the 24th International Joint Conference on
  Artificial Intelligence (IJCAI)}, pages 568--574. {AAAI} Press, 2015.

\bibitem{Lavi2008}
Ron Lavi, Ahuva Mu'alem, and Noam Nisan.
\newblock {Two simplified proofs for Roberts' theorem}.
\newblock {\em Social Choice and Welfare}, 32(3):407, 2008.

\bibitem{lavi2009truthful}
Ron Lavi and Chaitanya Swamy.
\newblock Truthful mechanism design for multidimensional scheduling via cycle
  monotonicity.
\newblock {\em Games and Economic Behavior}, 1(67):99--124, 2009.

\bibitem{LuYu08b}
Pinyan Lu and Changyuan Yu.
\newblock An improved randomized truthful mechanism for scheduling unrelated
  machines.
\newblock In {\em Proc. of the 25th Annual Symposium on Theoretical Aspects of
  Computer Science (STACS)}, volume~1 of {\em LIPIcs}, pages 527--538. Schloss
  Dagstuhl - Leibniz-Zentrum fuer Informatik, Germany, 2008.

\bibitem{LuYu08a}
Pinyan Lu and Changyuan Yu.
\newblock Randomized truthful mechanisms for scheduling unrelated machines.
\newblock In {\em Proc. of the 4th International Workshop on Internet and
  Network Economics (WINE)}, volume 5385 of {\em Lecture Notes in Computer
  Science}, pages 402--413, 2008.

\bibitem{random-mechanism}
Ahuva Mu'alem and Michael Schapira.
\newblock Setting lower bounds on truthfulness.
\newblock In {\em Proc. of the 18th Annual ACM Symposium on Discrete Algorithms
  (SODA)}, pages 1143--1152, 2007.

\bibitem{Mye81}
Roger~B. Myerson.
\newblock Optimal auction design.
\newblock {\em Mathematics of Operations Research}, 6:58--73, 1981.

\bibitem{nisan-ronen}
Noam Nisan and Amir Ronen.
\newblock Algorithmic mechanism design.
\newblock {\em Games and Economic Behavior}, 35(1-2):166--196, 2001.

\bibitem{Roberts}
Kevin Roberts.
\newblock Aggregation and revelation of preferences.
\newblock {\em The Characterization of Implementable Choice Rules}, pages
  321--348, 1979.

\bibitem{Roc87}
Jean-Charles Rochet.
\newblock {A necessary and sufficient condition for rationalizability in a
  quasi-linear context}.
\newblock {\em Journal of Mathematical Economics}, 16(2):191--200, 1987.

\bibitem{SakYu05}
Michael Saks and Lan Yu.
\newblock Weak monotonicity suffices for truthfulness on convex domains.
\newblock In {\em Proc. of the 6th ACM Conference on Electronic Commerce (EC)},
  pages 286--293. ACM, 2005.

\bibitem{Tsaggouris2009}
George Tsaggouris and Christos Zaroliagis.
\newblock Multiobjective optimization: Improved fptas for shortest paths and
  non-linear objectives with applications.
\newblock {\em Theory of Computing Systems}, 45(1):162--186, Jun 2009.

\bibitem{yu2009truthful}
Changyuan Yu.
\newblock Truthful mechanisms for two-range-values variant of unrelated
  scheduling.
\newblock {\em Theoretical Computer Science}, 410(21-23):2196--2206, 2009.

\end{thebibliography}

\newpage
\appendix
\section{Complexity considerations}\label{app:PTAS}
\newcommand{\Par}{\mathcal{P}}
In this section, we show that the optimization (non-strategic) version of the min-max path problem can be approximated arbitrarily well in polynomial time, for any constant number of agents. 

\begin{theorem}\label{thm:PTAS}
	For any constant number of agents, there exists a polynomial-time approximation scheme (PTAS) for the min-max path problem.
\end{theorem}

We employ a result by Tsaggouris et al.~\cite{Tsaggouris2009} on the  \emph{multi-objective} version of the shortest path problem. In this version, each edge $e$ has associated an $n$-dimensional \emph{vector} $w(e)=(w_1(e),\ldots,w_n(e))$. Each solution $x$ has an $n$-dimensional cost vector $c(x)=(c_1(x),\ldots,c_n(x))$ with $c_i(x)$ being the cost computed with respect to the $i^{th}$ component of weights $w(e)$,
\[
c_i(x) = \sum_{e \in x} w_i(s)\ . 
\]
One can see that our min-max path is a special case of this multi-objective optimization problem: For an edge $e$ owned by $i$, consider the vector 
\begin{align}\label{eq:multi-obj:weights}
w(e)=(0,\ldots,0,w_i(e),0,\ldots,0) && \text{ with } w_i(e)=t_i(e)
\end{align}
which then implies
\[
c_i(x) =  t_i(x)\ .  
\]

The set $\Par$ of Pareto solutions consists of all solutions that are not dominated: For every $x\in \Par$, there is no other solution $x'$ such that $c_i(x')\leq c_i(x)$ for all $i$, and one of these inequalities being strict.
\begin{definition}[$(1+\epsilon)$-Pareto set]
	 The set $\Par_{\epsilon}$ of  $(1+\epsilon)$-Pareto solutions is a set of solutions such that, for every $x \in \Par$, there is some $y$ in $\Par_{\epsilon}$ such that $c_i(y) \leq (1+\epsilon)c_i(x)$ for all $i$. 
\end{definition}

Let $\Par$ be the Pareto solutions and $\Par_\epsilon$ be the $(1+\epsilon)$-Pareto set. Let  $\nu$ and $m$ denote the number of nodes and edges in the graph, and $n$ be the number of agents as above. Another key parameter is the ratio between the maximum and minimum cost of an edge, i.e., 
\[
R_w = \max_e \max_{k,l} w_k(e)/w_l(e)\ . 
\]

\begin{theorem}[\cite{Tsaggouris2009}]\label{th:multi-obj:SP}
	$\Par_\epsilon$ can be computed in time $O\left(\nu m \left(\frac{\nu \log (\nu R_w)}{\epsilon}\right)^{n-1}\right)$.
\end{theorem}

Note that, in the above encoding of our problem \eqref{eq:multi-obj:weights}, $R$ is actually unbounded, which makes the above theorem of no use.  We shall therefore make a minor modifications of the weights, by setting the minimum weight to some  suitably small $\delta>0$ to be specified below:
\begin{align}\label{eq:multi-obj:weights-modified}
w'_i(e)=\max(\delta,w_i(e))
&& \text{ for all $i$ and $e$} \ .
\end{align}
(In particular, all $0$s of $w(e)$ are replaced by $\delta$ which is the minimum weight in the modified vectors.)
In order to specify $\delta$, we make the following observations. 
Let $SP(t)$ denote the length of the shortest path (sum of edge costs) for the original edge costs $t$. Let also $OPT(t)$ denote the optimum for the min-max cost for edge costs $t$. Then
\begin{equation}\label{eq:trivial-apx}
\frac{SP(t)}{n} \leq OPT(t) \leq SP(t). \ 
\end{equation}
(See Remark~\ref{rem:apx-general} below for a proof.)
If $SP(t)=0$ then this solution is also the optimum and we are done. Otherwise, we set $\delta:=\frac{\epsilon SP(t)}{n^2}$ and remove all edges $e$ whose cost $t_i(e)$, for $i$ being the agent owning edge $e$, is larger than $SP(t)$. Let $t'$ be the instance obtained from $t$ after this transformation. Now observe the following: 
\begin{itemize}
	\item Any solution that uses some discarded edge would cost at least $SP(t)$, so it will not be better than the shortest path. 
	\item Any solution $x$ which does not use any discarded edge, has a cost which is close in both edge weightings:
	 $$cost(x,t')\leq cost(x,t) + n \cdot \delta \leq  cost(x,t)+ \epsilon OPT(t)\ , $$
	where the first inequality is due to the fact that any simple path has at most $n$ edges, and we have increased each edge cost by at most $\delta$; the second inequality follows from \eqref{eq:trivial-apx} and by our choice of $\delta$.  
\end{itemize}
In particular, the set $\mathcal P'$ of Pareto solutions for the modified weights $w'$ contains a $(1+\epsilon)$-approximate solution for the original weights $w$. Since $R_{w'}\leq \frac{SP(t)}{\delta}= \frac{n^2}{\epsilon}$, we can apply Theorem~\ref{th:multi-obj:SP} and compute the $(1+\epsilon)$-Pareto set $\Par_{\epsilon}$ in polynomial time.
Then this set contains a polynomial number fo solutions. The solution $x$ in $\Par_{\epsilon}$ minimizing our cost function $cost(x,t')=\max_i c_i(x)$ is a $(1+\epsilon)$-approximation for the weights $w'$, and thus a $(1+\epsilon)^2$-approximation for the original weights $w$. That is, it is a $(1+\epsilon)^2$-approximation for the input $t$. Since we can choose $\epsilon>0$ arbitrarily small, this yields a PTAS. 

\section{Trivial $n$-approximation via VCG}
It is well known that every min-max problem admits a trivial $n$-approximate truthful VCG mechanism (see e.g. Nisan and Ronen \cite{nisan-ronen}). For convenience of the reader, we repeat here this simple folklore argument. Let $OPT(t)$ be the min-max optimum for costs $t$, and $SC(t)$ being the  social cost optimum (sum of agents costs) also with respect to $t$. 
By writing both these quantities according to the agents shares, we have
\[
OPT(t) = \max(OPT_1(t),\ldots,OPT_n(t))
\]
where $OPT_i(t)=t_i(x^*)$ for $x^*$ being the solution minimizing the min-max cost, and
\[
SC(t) = SC_1(t)+\cdots+SC_n(t)
\]
where $SC_i(t)=t_i(x)$ for $x$ being the solution minimizing the social cost. The solution $x$ minimizing the social cost (output by the VCG mechanism) has cost 
\begin{align*}
	cost(x,t) = \max_i t_i(x) = \max_i SC_i(t) &\leq \sum_i SC_i(t) \\
	&\leq \sum_i OPT_i(t)  \leq n \cdot \max_i OPT_i(t)
\end{align*}
which then implies that $x$ is an $n$-approximate solution. 

\begin{remark}\label{rem:apx-general}
	The optimum social cost and the min-max optimum are related according to following two inequalities:
	\begin{equation}\label{eq:n-apx:general}
	\frac{SC(t)}{n} \leq OPT(t) \leq SC(t)\ .
	\end{equation}
	The first one has been proved above. As for the second inequality, simply observe that $OPT(t)=\max_i(OPT_i(t)) \leq \max_i(SC_i(t))\leq SC(t)$, where the first inequality is due to the fact that $OPT(t)$ is the optimum for the min-max. 
		For the min-max path problem, $SC(t)$ is simply the length $SP(t)$ of the shortest path with respect to edge costs $t$. 
\end{remark}

\end{document}